\documentclass[aps,pra,superscriptaddress,twocolumn,nofootinbib]{revtex4-2}
\pdfoutput=1
\usepackage{dsfont}
\usepackage{comment}
\usepackage[utf8]{inputenc}
\usepackage[T1]{fontenc}
\usepackage[british]{babel}
\usepackage[dvipsnames,x11names]{xcolor}
\usepackage{amsmath,amssymb,amsthm,bm,amsfonts,bbm}
\usepackage{mathtools}
\usepackage[colorlinks=true,citecolor=Green,linkcolor=Purple,urlcolor=Blue]{hyperref}
\usepackage{physics}
\usepackage{extarrows}
\usepackage{enumitem}
\newcommand{\id}{\mathds{1}}

\newtheorem{theorem}{Theorem}
\newtheorem{proposition}{Proposition}

\newtheorem{corollary}{Corollary}

\newtheorem{task}{Task}

\begin{document}

\title{A lower bound on the classical simulation cost of star-network correlations}

%A lower bound on the classical simulation cost of star-network correlations
%Certifying the complexity of a qubit via joint measurements in a star network
%Classical simulation cost of star-network correlations grows exponentially with dimension
%A dimension- and party-dependent lower bound on simulating star-network correlations
%Joint measurements amplify the classical cost of simulating a qubit

\author{Martin J. Renner}
    %\email{martin.renner@icfo.eu}
    \affiliation{ICFO - Institut de Ciencies Fotoniques, The Barcelona Institute of Science and Technology, 08860 Castelldefels, Spain}

\date{\today}

\begin{abstract}
It is well established that quantum strategies outperform classical ones in several communication tasks. We study the quantum communication complexity of correlations arising from joint measurements on quantum systems distributed across a star network, where several parties each send a quantum system to a central node. We introduce an exclusion task that can be solved perfectly when each party sends a quantum $d$-level system, but would require a large classical message otherwise. In fact, the task cannot be solved with certainty if each of the $n$ parties sends a classical message with less than $n^{(d-1)}$ symbols. This implies an advantage of using quantum over classical messages in that scenario that scales with both, the dimension of the quantum system and the number of systems measured simultaneously. As an application, this shows that no finite-size classical description of a qubit suffices to reproduce the statistics of a joint measurement on sufficiently many qubits.
%even when shared randomness is available
\end{abstract}

\maketitle

\section{Introduction}

Harnessing quantum resources to outperform classical communication strategies is one of the central themes of quantum information theory~\cite{Brassard2003, Buhrman10_review}. Famous examples include protocols such as Random Access Coding~\cite{WiesnerRAC1983} and Superdense Coding~\cite{bennettdensecoding}, which show that a single quantum system can carry more (or more useful) information than its classical counterpart. A particularly striking instance of such an advantage arises in the bipartite setting of a prepare-and-measure scenario, where a sender transmits a quantum state and a receiver performs a measurement on the received state. In this setting, it is known that a quantum message of dimension $d$ can outperform any classical message of size scaling less than exponentially in $d$~\cite{raz1999, Buhrman2001PRL, hiddenmatching, Gavinsky2007}.

A complementary question to finding tasks with a large quantum-classical gap is to ask how much classical communication is required to simulate the statistics of quantum correlations exactly~\cite{Brassard1999, tonerbacon2003, Degorre2005, Montina2011KSmodel, Renner2023, Renner2023b, sidajaya23, Naik2025, Schlosser2026}. Here the bipartite answer is remarkably economical: when a single sender transmits a qubit ($d=2$) to a receiver who may then perform an arbitrary measurement, all resulting correlations can be reproduced by a classical message of only two bits together with shared randomness~\cite{tonerbacon2003, Renner2023}. It is natural to ask whether this remains true once several quantum systems are measured jointly by a common receiver, or whether joint measurements can change the picture.

In this work, we address this question by introducing a simple exclusion task on a star network with $n$ senders and one receiver. The task can be won perfectly when each sender transmits a $d$-level quantum system, whereas a perfect classical winning probability cannot be achieved if each party sends a classical message of less than $n^{(d-1)}$ symbols. This shows an exponential scaling in $d$ already for $n=2$. Furthermore, this means that no classical message of fixed size can reproduce the statistics of a qubit once it is measured jointly with sufficiently many other qubits. Our construction also provides an example of an exponential quantum advantage in a multipartite communication scenario~\cite{Bowles2015b, Doolittle2026, Pandit2026, Chakraborty2026b}, a setting that has received less attention than the bipartite case.
%This implies that no classical message of any fixed size can perfectly simulate the statistics of a qubit once it is measured jointly with sufficiently many other qubits.%This shows that joint measurements are a genuine resource for amplifying the quantum-classical communication gap.

To construct this task we use the concept of quantum state antidistinguishability, also known as quantum state exclusion. This was first introduced by Caves, Fuchs, and Schack~\cite{Caves2002} and is the basis of the well-known Pusey, Barrett, Rudolph (PBR) theorem, which challenges the epistemic view of quantum mechanics~\cite{PBR2012}. In recent years, antidistinguishability~\cite{Jain2014, Heinosaari2018, Russo2023, Johnston2025} has been recognised as a useful tool for demonstrating quantum-to-classical advantages in communication tasks~\cite{Perry2015, Heinosaari2019, Havlicek2020, Bae2025}. In particular, this work is inspired by the work of Havlíček and Barrett~\cite{Havlicek2020}.

\begin{figure}[]
    \centering
    \includegraphics[width=0.9\linewidth]{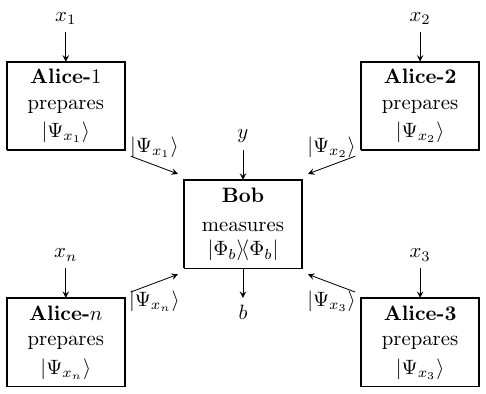}
    \caption{The setup: Several Alices receive an input and prepare a quantum system that they send to Bob. Bob can measure all received quantum systems together and produces an outcome according to his measurement result.
    %In the analogous classical scenario, each Alice sends a classical message to Bob, who determines his outcome using this classical information.
    }
    \label{figuresetup}
\end{figure}

\begin{figure*}[]
    \centering
    \includegraphics[width=1.0\linewidth]{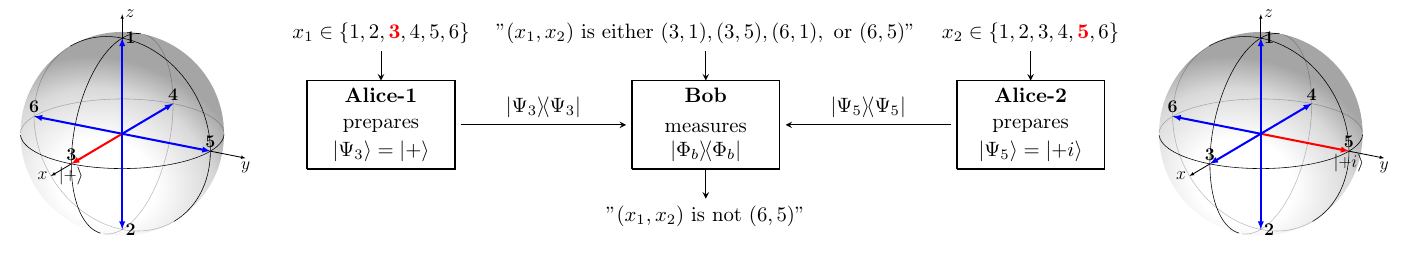}
    \caption{Each Alice receives a number between 1 and 6. Given the input, they prepare one of the eigenstates of the three Pauli operators and send it to Bob. Here, $x_1=3$ and $x_2=5$, so Alice-1 sends the state $\ket{+}$ and Alice-2 the state $\ket{+i}$. Bob receives a promise about the inputs of Alice-1 and Alice-2. For instance, the referee tells him that $x_1$ is either $3$ or $6$ and $x_2$ is either $1$ or $5$. They win if Bob can exclude one of the possible combinations. So he is supposed to reply either ``$(x_1,x_2)$ is not $(3,1)$'' or ``$(x_1,x_2)$ is not $(6,1)$'' or ``$(x_1,x_2)$ is not $(6,5)$''. We show that the task can be solved perfectly with the presented qubit strategy but requires a classical message of more than two bits per qubit.}
    \label{figure1}
\end{figure*}

\section{Preliminaries}

In this work, we consider a star-shaped network in which several parties, called Alice-$i$, can send a quantum system to a central node called Bob (see Fig.~\ref{figuresetup}). Bob receives all states and can measure them together with a joint measurement. The quantum states that Alice can send are formally described as positive operators $\rho \in L(\mathbb{C}^d)$ with $\Tr[\rho]=1$. In this work, we focus on pure states $\ket{\Psi}\in \mathbb{C}^d$, where $\rho=\ketbra{\Psi}$. The measurements that Bob can perform are in general positive operator-valued measurements (POVMs) $M_{b|y}\geq 0$ with $\sum_b M_{b|y}=\id_d$. In the considered scenario, the resulting correlations become (due to Born's rule):
\begin{align}
    p_Q(b|x_1, x_2, ..., x_n, y)=\Tr[M_{b|y} \ (\rho_{x_1}\otimes\rho_{x_2}\otimes ... \otimes \rho_{x_n})] \, .
\end{align}
Throughout this work, we assume the parties do not share additional entanglement and we assume that the Alices only communicate with Bob and not among each other.

In the analogous classical scenario, all parties can send a classical message to the central node. In general, the parties can share some global randomness $\lambda$, and the classical message that each Alice sends to Bob may depend on her input and on $\lambda$, i.e.\ $p_{A_i}(c_i|x_i, \lambda)$. Bob receives all messages $c_i$ and produces an outcome according to the distribution $p_B(b|c_1, c_2, \ldots, c_n, y, \lambda)$. In total, the outcome $b$ obeys:
\begin{align}
\begin{split}
    p_C&(b|x_1, x_2, ..., x_n, y)=\\
    &\sum_{\lambda, c_1, c_2, ..., c_n } \left(\prod_{i=1}^n p_{A_i}(c_i|x_i, \lambda) \right)\cdot p_B(b|c_1, c_2, ..., c_n, y, \lambda)
\end{split}
\end{align}
In this work, we are interested in how much classical communication is (at least) required to substitute a quantum system of a given dimension. It is known that in the case of $n=1$ a qubit can be replaced by a classical message of two bits and shared randomness~\cite{tonerbacon2003, Renner2023}. Here we show that the communication cost must scale with both the dimension and the number of involved parties.

A key concept used throughout this work is that of \emph{antidistinguishability}. A set of states $\{\ket{\Psi_1}, \ket{\Psi_2}, \ldots, \ket{\Psi_k}\}$ is called antidistinguishable if there exists a POVM $\{M_1, M_2, \ldots, M_k\}$ such that $\Tr[M_i \ketbra{\Psi_i}]=0$ for all $i$. In other words, when measuring such a set, each outcome $i$ certifies that the state was \emph{not} $\ket{\Psi_i}$. 

\section{The task}

In order to separate the classical and quantum regimes, we introduce the following exclusion task:
\begin{task}\label{task}
    Each Alice (labelled by index $i$) receives an input $x_i\in \{1, 2, \ldots, m\}$. At the same time, Bob receives an input $(\tilde{x}^0_i, \tilde{x}^1_i)$ (with $\tilde{x}^0_i\neq \tilde{x}^1_i$) for each Alice, and it is promised that either $\tilde{x}^0_i=x_i$ or $\tilde{x}^1_i=x_i$. Finally, Bob is asked to output one combination $(b_1, b_2, \dots, b_n)$ such that either $b_i=\tilde{x}^0_i$ or $b_i=\tilde{x}^1_i$ and they win if and only if $(b_1, b_2, \dots, b_n)\neq (x_1, x_2, \dots, x_n)$, meaning that $b_i \neq x_i$ for at least one $i$.
\end{task}
Although the task can in principle be defined for any $m$, we will focus on those instances in which the task can be won perfectly with a quantum strategy. More precisely, we denote by $\bar{m}(n,d)$ the maximum value such that a perfect quantum strategy for $n$ parties sending $d$-dimensional quantum states exists.
%It is useful to write the winning probability in the following way:

%Here, $m$ is a number that depends on the number of involved parties and the dimension of the quantum system sent, hence we write . We will show that this task can be solved perfectly if each Alice sends a quantum $d$-level system to Bob, but it cannot be solved if each Alice sends fewer than $\log_2(m)$ classical bits (even in the presence of shared randomness). The winning probability of a strategy is defined as:
%\begin{align}
%    p_{\mathrm{win}}=\sum_{x_i, y} \Bigl(1-p\Bigl(\bigwedge_{i=1}^{n}(b_i=x_i)\Bigr)\Bigr). \label{probwin}
%\end{align}

%(Note to myself: Maybe example first and this part afterwards.)
The connection to antidistinguishability is the following. Given Bob's input $(\tilde{x}^0_i, \tilde{x}^1_i)$ for each Alice $i$, Bob knows that the state of the $n$ quantum systems is one of the $2^n$ possible tensor product states:
\begin{align}
    \ket*{\Psi_{\vec{r}}} := \ket*{\Psi_{\tilde{x}^{r_1}_1}} \otimes \ket*{\Psi_{\tilde{x}^{r_2}_2}} \otimes \cdots \otimes \ket*{\Psi_{\tilde{x}^{r_n}_n}}, \quad \vec{r} \in \{0,1\}^n.
\end{align}
The task asks Bob to output some $\vec{r}$ such that $\ket*{\Psi_{\vec{r}}}$ was definitely \emph{not} the prepared state. This is precisely the definition of antidistinguishability: if the set $\{\ket*{\Psi_{\vec{r}}}\}_{\vec{r}}$ is antidistinguishable, then for every possible input $(x_1,\ldots,x_n)$ there exists a measurement outcome that certifies which combination was not prepared. If Bob outputs that outcome, they can solve Task~\ref{task} perfectly.
%Conversely, if the set is not antidistinguishable for some input, there is no measurement that can exclude any single combination with certainty.
Hence, the task can be solved perfectly if for every possible input, the set of $2^n$ states $\{\ket*{\Psi_{\vec{r}}}\}_{\vec{r}}$ is antidistinguishable.

\section{Simplest example}

We think it is best to explain the task with an example and discuss the general case afterwards. Suppose there are two Alices ($n=2$) that are both allowed to send a qubit message to Bob ($d=2$). We choose $m=6$, so each Alice receives a number $x_i$ between $1$ and $6$. Given the input, Alice-$i$ prepares one of the six states (see also Fig.~\ref{figure1}):
\begin{align*}
&\ket{0}\ (x=1),\quad \ket{1}\ (x=2),\quad \ket{+}\ (x=3),\\
&\ket{-}\ (x=4),\quad \ket{+i}\ (x=5),\quad \ket{-i}\ (x=6),
\end{align*}
where $\ket{\pm}=(\ket{0}\pm\ket{1})/\sqrt{2}$ and $\ket{\pm i}=(\ket{0}\pm i\ket{1})/\sqrt{2}$. Alice-2 follows the same strategy.

Suppose for instance that the first Alice gets $x_1=3$ and sends $\ket{+}$, while the second Alice receives $x_2=1$ and sends $\ket{0}$. Suppose further that Bob is told that $x_1 \in\{1,3\}$ and $x_2\in\{1,3\}$, so the pair $(x_1,x_2)$ is either $(1,1),(1,3),(3,1)$, or $(3,3)$. His goal is to output one combination that is definitely not $(x_1,x_2)$, i.e.\ any combination except $(3,1)$ is a valid answer. In the quantum strategy, Bob can achieve this by measuring the two qubits together in the entangled basis:
\begin{align}
    \ket{\Phi_1}&=(\ket{01}+\ket{10})/\sqrt{2}\\
    \ket{\Phi_2}&=(\ket{0-}+\ket{1+})/\sqrt{2}\\
    \ket{\Phi_3}&=(\ket{+1}+\ket{-0})/\sqrt{2}\\
    \ket{\Phi_4}&=(\ket{+-}+\ket{-+})/\sqrt{2} \, .
\end{align}
This basis coincides with the one used in the PBR argument~\cite{PBR2012}. One can verify that $\sum_i \ketbra{\Phi_i}=\id_4$ and that:
\begin{align}
    \braket{\Phi_1}{00}=\braket{\Phi_2}{0+}=
    \braket{\Phi_3}{+0}=\braket{\Phi_4}{++}=0 \, .
\end{align}
Hence, if Bob's measurement outcome is $\ket{\Phi_1}$, he knows the preparation was certainly not $\ket{00}$ and therefore $(x_1,x_2)\neq (1,1)$. Hence, Bob can output $b=(1,1)$ to certainly win the task. Similarly, he outputs $b=(1,3)$ if the outcome is $\ket{\Phi_2}$, $b=(3,1)$ if $\ket{\Phi_3}$, and $b=(3,3)$ if $\ket{\Phi_4}$.

If Bob receives a different promise, he can always find another measurement basis to solve the task. If both possible states of one Alice subtend a $90^\circ$ angle on the Bloch sphere, he first applies local unitaries mapping those two states to $\ket{0}$ and $\ket{+}$ respectively, and then performs the same entangled measurement as above. If the two possible states for one Alice-$i$ are orthogonal, he can simply measure that qubit in the corresponding basis. Doing so, he learns the value of $x_i$ and can simply output any combination with $b_i\neq x_i$.

\section{Classical strategies}

The same task cannot be solved with a classical message of two bits from each Alice. Suppose each Alice is allowed to send a message to Bob that does not allow her to encode her input perfectly (e.g.\ two classical bits). Then each Alice-$i$ necessarily sends the same message $c_i$ for at least two different inputs. If Bob's promise is exactly that $x_i$ is one of those two values, then the message $c_i$ carries no information to distinguish them.

To give a concrete example, suppose each Alice sends $c=00$ for $x=1$; $c=01$ for $x=2$; $c=10$ for $x\in\{3,4\}$; and $c=11$ for $x\in\{5,6\}$. If Bob is told that $(x_1,x_2)\in\{(3,5),(3,6),(4,5),(4,6)\}$, he receives the same messages $c_1=10$ and $c_2=11$ for all four combinations and cannot exclude any. The same logic applies to an arbitrary number of parties: if no Alice is allowed to send a message of at least $m$ symbols, each Alice must send the same message for at least two of her inputs, and there is always a promise for which Bob cannot exclude any combination. Hence, to solve the task perfectly with a classical message, at least one Alice must send a message of $m$ classical symbols.

\begin{theorem}\label{thm:classical}
    If each party sends a classical message of fewer than $m$ symbols, then Task~\ref{task} cannot be won with certainty, even in the presence of shared randomness.
\end{theorem}
A detailed proof is given in Appendix~\ref{classicalstrategies}, where we also analyse classical strategies quantitatively and compute winning probabilities for a message of a given length. We do not want to hide the fact that although a large classical message is required for a \emph{perfect} solution, a shorter classical message can already perform quite well for this task. Even in the extreme case where no Alice sends any information, they can win with probability $1-2^{-n}$ if Bob simply guesses a random outcome.

\section{More qubits and larger dimensions}

%We now determine under which circumstances the task can be solved perfectly with a quantum strategy. 
We can now describe a general quantum strategy that is capable of solving the introduced task perfectly. Suppose there are $n$ Alices and each receives an input $x_i\in \{1,2,\ldots,m\}$. Given that input, Alice-$i$ prepares a $d$-dimensional state $\ket{\Psi_{x_i}}$ and sends it to Bob. For reasons that will become clear below, these states should satisfy:
\begin{align}
    \forall\, j\neq k\in \{1,2,\ldots,m\}:\quad |\braket{\Psi_{j}}{\Psi_{k}}|\leq f(n), \label{sphericalcode}
\end{align}
for a given $f(n)$ that will be determined soon. This is known as a \emph{complex spherical code}~\cite{DelsarteGoethalsSeidel1977, ConwaySloane1999, Roy2014}. We will see below that if, in a given dimension $d$, $m$ different states satisfying~\eqref{sphericalcode} can be found, they can solve Task~\ref{task} perfectly. To see this, note that Bob learns that the input of Alice-$i$ is either $\tilde{x}^0_i$ or $\tilde{x}^1_i$. Hence, he knows that the combined system of the states he received is in one of the $2^n$ states
\begin{align}
    \ket*{\Psi_{\vec{r}}} := \ket*{\Psi_{\tilde{x}^{r_1}_1}} \otimes \ket*{\Psi_{\tilde{x}^{r_2}_2}} \otimes \cdots \otimes \ket*{\Psi_{\tilde{x}^{r_n}_n}}\, , \quad \vec{r} \in \{0,1\}^n \, . \label{eq:prodstates}
\end{align}
His task is to output one combination that was definitely not the prepared one, which is equivalent to antidistinguishing these states. 

It can be established that for a certain value of $f(n)$ the set of $2^n$ states $\{\ket*{\Psi_{\vec{r}}}\}$ are in fact always antidistinguishable. A simple but weaker bound can be shown using a Theorem by Johnston, Russo, and Sikora in Ref.~\cite{Johnston2025} (see Appendix~\ref{app:johnston}). Here, we follow the techniques of Ref.~\cite{PBR2012} and Ref.~\cite{Jain2014} instead. In that case, we define $f(n)$ as:
\begin{align}
    f(n):=\frac{1-(\sqrt[n]{2}-1)^2}{1+(\sqrt[n]{2}-1)^2} \, . \label{deffunction}
\end{align}
We further define the critical angle $\alpha^c_n$ via $\cos(2\alpha^c_n):=f(n)$. The strategy to antidistinguish those states, and therefore the strategy to win the task, goes as follows. Given their input $x_i$, each Alice prepares the state $\ket{\Psi_{x_i}}$ and sends it to Bob. Bob learns that the input of Alice-$i$ is either $\tilde{x}^0_i$ or $\tilde{x}^1_i$. Hence, he knows that the state he holds is either $\ket*{\Psi_{\tilde{x}^0_i}}$ or $\ket*{\Psi_{\tilde{x}^1_i}}$. In a sequence of steps (see Appendix~\ref{quantumstrategies}), he applies a transformation to map the two states into:
\begin{align}
    \ket*{\Psi_{\tilde{x}^0_i}}&\mapsto \ket{\psi_0}=\cos{(\alpha^c_n)}\ket{0}+\sin{(\alpha^c_n)}\ket{1}\\
    \ket*{\Psi_{\tilde{x}^1_i}}&\mapsto \ket{\psi_1}=\cos{(\alpha^c_n)}\ket{0}-\sin{(\alpha^c_n)}\ket{1}
\end{align}
We show that such a transformation is possible whenever $|\braket*{\Psi_{\tilde{x}^0_i}}{\Psi_{\tilde{x}^1_i}}|\leq f(n)$ which is guaranteed by~\eqref{sphericalcode}.

He applies the same procedure to the system received from each Alice. Afterwards, Bob knows that the total state is one of the $2^n$ possible states labelled by $\vec{r}\in\{0,1\}^n$:
\begin{align}
    \ket{\Omega_{\vec{r}}}:=&\ket{\psi_{r_1}}\otimes \ket{\psi_{r_2}} \otimes ... \otimes \ket{\psi_{r_n}}\\
    =&\sum_{\vec{z}} (-1)^{\vec{r}\cdot \vec{z}}\cos(\alpha^c_n)^{(n-\vec{z}\cdot \vec{1})}\sin(\alpha^c_n)^{\vec{z}\cdot \vec{1}}\ket{\vec{z}} \, .
\end{align}
where $\vec{z}\in\{0,1\}^n$, $\ket{\vec{z}}:=\ket{z_1}\otimes \dots \otimes \ket{z_n}$, and $\vec{1}:=(1,1,...,1)$. Crucially, this set of $2^n$ states is antidistinguishable: there exists a measurement such that outcome $\vec{r}$ occurs with zero probability for the preparation $\ket{\Omega_{\vec{r}}}$.
%Since Bob knows that the actual preparation is one of these $2^n$ states, antidistinguishability is exactly what is needed to solve Task~\ref{task} perfectly. In fact, each outcome tells Bob which combination was definitely \emph{not} prepared.
Following the original PBR argument~\cite{PBR2012} or Ref.~\cite{Jain2014}, the antidistinguishing measurement basis is:
\begin{align}
    \ket*{\Phi_{\vec{r}}}=\frac{1}{\sqrt{2^n}}\left( \ket*{\vec{0}} -\sum_{\vec{z}\neq \vec{0}} (-1)^{\vec{r}\cdot \vec{z}} \ket{\vec{z}} \right).
\end{align}
A direct calculation (see App.~\ref{quantumstrategies}) shows that for all $\vec{r}$:
\begin{align}
    \sqrt{2^n}\braket{\Omega_{\vec{r}}}{\Phi_{\vec{r}}}=\cos{(\alpha^c_n)}^n(2-(1+\tan{(\alpha^c_n)})^{n})=0 \, ,
\end{align}
where we use that $f(n)=\cos(2\alpha^c_n)$ is chosen such that $\tan{(\alpha^c_n)}=\sqrt[n]{2}-1$. Hence, when Bob performs this measurement and receives outcome $\ket{\Phi_{\vec{r}}}$, he knows for sure that the preparation was not $\ket{\Omega_{\vec{r}}}$ and can solve the task perfectly. This establishes the following:
\begin{theorem}\label{thm:quantum}
    Let $n$ be the number of parties. If there exist $m$ $d$-dimensional quantum states $\{\ket{\Psi_1},\ldots,\ket{\Psi_m}\}$ satisfying 
    $$|\braket{\Psi_j}{\Psi_k}|\leq \frac{1-(\sqrt[n]{2}-1)^2}{1+(\sqrt[n]{2}-1)^2}$$ for all $j\neq k$, then Task~\ref{task} can be solved perfectly by a quantum strategy in which each Alice sends a $d$-dimensional quantum system.
\end{theorem}

\section{Spherical codes}

We have shown that the communication complexity of joint measurements on a $d$-level quantum system is connected to the largest complex spherical code satisfying~\eqref{sphericalcode}. We use this to determine some lower bounds on $\bar{m}(n,d)$. It is useful to first note some values of $f(n)$. For $n=1$ we obtain $f(1)=0$, so the states must be mutually orthogonal and $\bar{m}(1,d)=d$. For $n=1$, the task matches the one introduced in Ref.~\cite{Naik2022}. For $n=2$ we obtain $f(2)=1/\sqrt{2}$.

For qubits ($d=2$), pure states can be represented via the Bloch sphere: $\ketbra{\Psi}=(\id + \vec{v}\cdot \vec{\sigma})/2$ where $\vec{v}\in \mathbb{R}^3$ denotes the Bloch vector ($|\vec{v}|=1$ for pure states) and $\vec{\sigma}=(\sigma_x,\sigma_y,\sigma_z)$ denote the standard Pauli matrices. The inner product becomes $|\braket{\Psi_k}{\Psi_j}|^2=(1+\vec{v}_k\cdot\vec{v}_j)/2$, so condition~\eqref{sphericalcode} becomes:
\begin{align}
    \vec{v}_k\cdot\vec{v}_j\leq 2\cos(2\alpha^c_n)^2-1=\cos(4\alpha^c_n). \label{sphericalcodequbit}
\end{align}
For $n=2$, this requires $\vec{v}_k\cdot\vec{v}_j\leq 0$, which is realised by the six unit vectors $\pm \vec{e}_x,\pm \vec{e}_y,\pm \vec{e}_z$ used in the first example (see Fig.~\ref{figure1}). For $n=3$, pairs of Bloch vectors must subtend an angle of at least $58.3^\circ$. Ref.~\cite{Musin2010} shows that 12 such vectors exist (e.g.\ the vertices of a regular icosahedron) but 13 do not, so joint measurements on three qubits cannot be simulated with fewer than $\log_2(12)>3$ bits per qubit. For $n=4$, 24 points suffice~\cite{Robinson1961}. Further values are listed in Ref.~\cite{Sloanlist, Listspherecode} (see also Tab.~\ref{tab:lowerbounds}).

\begin{table}[h]
    \centering
\begin{tabular}{c||c|c|c|c|c|c}
$n$   & 1 & 2 & 3 & 4 & 5 & 6 \\\hline\hline
$f(n)$ & $0$ & $1/\sqrt{2}$ & $\approx0.873$ & $\approx0.931$ & $\approx0.957$ & $\approx0.970$ \\\hline\hline
$4\alpha_n^c$ & $180^\circ$ & $90^\circ$ & $\approx58.3^\circ$ & $\approx42.9^\circ$ & $\approx33.8^\circ$ & $\approx27.9^\circ$ \\\hline \hline
$\bar{m}(n,d=2)$ & 2 & 6 & 12 & 24 & 38 & 56 \\
\end{tabular}
    \caption{Values of $f(n)$, the critical Bloch angle $4\alpha_n^c$ (note that $\cos(2\alpha^c_n):=f(n)$), and the optimal number of codewords $\bar{m}(n,d=2)$ (from Ref.~\cite{Sloanlist, Listspherecode}).
    %Values for $d\geq 3$ depend on complex spherical codes that are not fully characterised in general; see the lower bounds in the text.
    }
    \label{tab:lowerbounds}
\end{table}

\subsection{Spherical codes in higher dimensions}

For $d\geq 3$, one can derive a lower bound on $\bar{m}(n,d)$ via a volumetric argument (see also Ref.~\cite{Havlicek2020, Montina2011}). Starting from a random state $\ket{\Psi_1}$, we iteratively exclude all states $\ket{\Psi}$ with $|\braket{\Psi_k}{\Psi}|\geq f(n)$ and pick the next state from those remaining, until no more states can be added. By construction, the resulting set is a valid complex spherical code, and its size is at least the ratio of the total sphere volume to the volume of a spherical cap with angle $2\alpha^c_n$. In Appendix~\ref{appspherecode}, we show:
\begin{align}
    \frac{S_{2\alpha^c_n}}{S}=(1-\cos{(2\alpha^c_n)}^2)^{d-1},
\end{align}
from which we obtain:
\begin{align}
    \bar{m}(n,d)\geq \left( \frac{1+(\sqrt[n]{2}-1)^2}{2(\sqrt[n]{2}-1)}\right)^{2(d-1)}\, .
\end{align}
For $n=2$ this gives $\bar{m}(2,d)\geq 2^{d-1}$, an exponential scaling with dimension. It should be stressed that this can be achieved with only four-outcome measurements (or $2^n$ in general), while most communication tasks require measurements with more outcomes (Refs.~\cite{Montina2011,Havlicek2020} are exceptions). For $n=3$, one obtains $\bar{m}(3,d)\geq 4^{d-1}$. For general $n$, a simpler but weaker bound (derived in Appendix~\ref{appspherecode}) is:
\begin{align}
    \bar{m}(n,d)\geq \left(\frac{n^2}{4}\right)^{d-1}.
\end{align}
In particular, $\bar{m}(n,d)\geq n^{d-1}$: for $n=2$ and $n=3$ this follows from the explicit bounds $2^{d-1}$ and $4^{d-1}$ above, while for $n\geq 4$ we have $n^2/4\geq n$. Combined with Theorems~\ref{thm:classical} and~\ref{thm:quantum}, this shows that any classical simulation with less than $n^{d-1}$ (or equivalently $(d-1)\log_2(n)$ classical bits) per party is insufficient to simulate a joint measurement on the $n$ $d$-dimensional systems exactly.

\section{Conclusion}
 
We have shown that the classical communication cost of simulating joint measurements on $n$ quantum $d$-level systems grows with both $n$ and $d$, with the required number of messages growing exponentially in $d$ already for $n=2$. In particular, no finite classical description of a qubit can reproduce the statistics of joint measurements on arbitrarily many qubits, in sharp contrast to the prepare-and-measure case, where a qubit is always simulable by two classical bits and shared randomness~\cite{tonerbacon2003, Renner2023}.
%This gap closes only if the parties share entanglement with Bob: remote state preparation then requires just two bits per party, recovering the single-qubit cost and highlighting entanglement as the resource responsible for the separation.
This shows that any classical model of the qubit that is simulable with a bounded message, for instance the Kochen - Specker model~\cite{KochenSpecker1967, Montina2011KSmodel} or the model of Ref.~\cite{Renner2023}, necessarily fails to reproduce the statistics of joint measurements on two or more qubits.
%Likewise, any model describing each qubit by a fixed number of classical bits, such as the quantum simulation logic of Ref.~\cite{Johansson2019}, cannot exactly reproduce general quantum circuits involving joint measurements, since the required number of bits must grow with the number of qubits.

Several open questions remain. Our results are not very robust to noise, since a smaller classical message can achieve already a quite high score for the task we introduce. Finding other tasks that avoid this problem is of particular interest.
%From a given perspective, this is unsurprising since a finite classical message can be used to approximate the Bloch sphere via some polytop and therefore simulate this task by using for instance similar techniques as in Ref.~\cite{Bowles2015}.
It would be also worth exploring whether a similar scaling persists when Bob has no input, where a single $d$-dimensional sender can already be perfectly simulated by $d$ classical symbols~\cite{FrenkelWeiner2015}, while two or three senders require more than one bit per qubit~\cite{Chakraborty2023}.\\

\section*{Acknowledgements}
We thank Antonio Ac\'in, Mir Alimuddin, Alexander Bernal, Edwin Peter Lobo, and Armin Tavakoli for insightful discussions and comments on the manuscript. We acknowledge financial support through the Juan de la Cierva postdoctoral fellowship (La ayuda JDC2024-055405-I, financiada por MICIU/AEI/10.13039/501100011033 y por el FSE+.), the Government of Spain (Severo Ochoa CEX2019-000910-S and FUNQIP), Fundaci\'o Cellex, Fundaci\'o Mir-Puig, Generalitat de Catalunya (CERCA program) and the European Union (NEQST 101080086).

%\nocite{apsrev42Control}
%\bibliographystyle{0_MTQ_apsrev4-2_corrected.bst}
\bibliography{bib.bib}

@ARTICLE{Buhrman10_review,
       author = {{Buhrman}, Harry and {Cleve}, Richard and {Massar}, Serge and {de Wolf}, Ronald},
        title = "{Nonlocality and communication complexity}",
      journal = {Reviews of Modern Physics},
         year = 2010,
        month = jan,
       volume = {82},
       number = {1},
        pages = {665-698},
          doi = {10.1103/RevModPhys.82.665},
archivePrefix = {arXiv},
       eprint = {0907.3584},
 primaryClass = {quant-ph},
       adsurl = {https://ui.adsabs.harvard.edu/abs/2010RvMP...82..665B}
}

@ARTICLE{PBR2012,
       author = {{Pusey}, Matthew F. and {Barrett}, Jonathan and {Rudolph}, Terry},
        title = "{On the reality of the quantum state}",
      journal = {Nature Physics},
         year = 2012,
        month = jun,
       volume = {8},
       number = {6},
        pages = {476-479},
          doi = {10.1038/nphys2309},
archivePrefix = {arXiv},
       eprint = {1111.3328},
 primaryClass = {quant-ph},
       adsurl = {https://ui.adsabs.harvard.edu/abs/2012NatPh...8..476P}
}

@ARTICLE{tonerbacon2003,
       author = {{Toner}, B.~F. and {Bacon}, D.},
        title = "{Communication Cost of Simulating Bell Correlations}",
      journal = {Physical Review Letters},
         year = 2003,
        month = oct,
       volume = {91},
       number = {18},
          eid = {187904},
        pages = {187904},
          doi = {10.1103/PhysRevLett.91.187904},
archivePrefix = {arXiv},
       eprint = {quant-ph/0304076},
 primaryClass = {quant-ph},
       adsurl = {https://ui.adsabs.harvard.edu/abs/2003PhRvL..91r7904T}
}

@ARTICLE{degorre2005,
       author = {{Degorre}, Julien and {Laplante}, Sophie and {Roland}, J{\'e}r{\'e}mie},
        title = "{Simulating quantum correlations as a distributed sampling problem}",
      journal = {Physical Review A},
         year = 2005,
        month = dec,
       volume = {72},
       number = {6},
          eid = {062314},
        pages = {062314},
          doi = {10.1103/PhysRevA.72.062314},
archivePrefix = {arXiv},
       eprint = {quant-ph/0507120},
 primaryClass = {quant-ph},
       adsurl = {https://ui.adsabs.harvard.edu/abs/2005PhRvA..72f2314D}
}

@ARTICLE{Brassard1999,
       author = {{Brassard}, Gilles and {Cleve}, Richard and {Tapp}, Alain},
        title = "{Cost of Exactly Simulating Quantum Entanglement with Classical Communication}",
      journal = {Physical Review Letters},
         year = 1999,
        month = aug,
       volume = {83},
       number = {9},
        pages = {1874-1877},
          doi = {10.1103/PhysRevLett.83.1874},
archivePrefix = {arXiv},
       eprint = {quant-ph/9901035},
 primaryClass = {quant-ph},
       adsurl = {https://ui.adsabs.harvard.edu/abs/1999PhRvL..83.1874B}
}

@article{Brassard2003,
    doi = {10.1023/A:1026009100467},
	year = {2003},
	month = {Nov},
	publisher = {Springer},
	volume = {33},
	number = {11},
	pages = {1593--1616},
	author = {Gilles Brassard},
	title = {Quantum Communication Complexity},
	journal = {Foundations of Physics}
}

@article{bennettdensecoding,
  title = {Communication via one- and two-particle operators on Einstein-Podolsky-Rosen states},
  author = {Bennett, Charles H. and Wiesner, Stephen J.},
  journal = {Physical Review Letters},
  volume = {69},
  issue = {20},
  pages = {2881--2884},
  numpages = {0},
  year = {1992},
  month = {Nov},
  publisher = {American Physical Society},
  doi = {10.1103/PhysRevLett.69.2881},
  url = {https://link.aps.org/doi/10.1103/PhysRevLett.69.2881}
}

@inproceedings{hiddenmatching,
author = {Bar-Yossef, Ziv and Jayram, T. S. and Kerenidis, Iordanis},
title = {Exponential Separation of Quantum and Classical One-Way Communication Complexity},
year = {2004},
isbn = {1581138520},
publisher = {Association for Computing Machinery},
address = {New York, NY, USA},
url = {https://doi.org/10.1145/1007352.1007379},
doi = {10.1145/1007352.1007379},
booktitle = {Proceedings of the Thirty-Sixth Annual ACM Symposium on Theory of Computing},
pages = {128–137},
numpages = {10},
location = {Chicago, IL, USA},
series = {STOC '04}
}

@article{WiesnerRAC1983,
author = {Wiesner, Stephen},
title = {Conjugate Coding},
year = {1983},
issue_date = {Winter-Spring 1983},
publisher = {Association for Computing Machinery},
address = {New York, NY, USA},
volume = {15},
number = {1},
issn = {0163-5700},
url = {https://doi.org/10.1145/1008908.1008920},
doi = {10.1145/1008908.1008920},
journal = {SIGACT News},
month = {Jan},
pages = {78–88},
numpages = {11}
}

@inproceedings{raz1999,
  title={Exponential separation of quantum and classical communication complexity},
  author={Raz, Ran},
  booktitle={Proceedings of the thirty-first annual ACM symposium on Theory of computing},
  pages={358-367},
  doi={https://doi.org/10.1145/301250.301343},
  year={1999}
}

@ARTICLE{Renner2023,
       author = {{Renner}, Martin J. and {Tavakoli}, Armin and {Quintino}, Marco T{\'u}lio},
        title = "{Classical Cost of Transmitting a Qubit}",
      journal = {Physical Review Letters},
         year = 2023,
        month = {March},
       volume = {130},
       number = {12},
          eid = {120801},
        pages = {120801},
          doi = {10.1103/PhysRevLett.130.120801},
archivePrefix = {arXiv},
       eprint = {2207.02244},
 primaryClass = {quant-ph},
       adsurl = {https://ui.adsabs.harvard.edu/abs/2023PhRvL.130l0801R}
}

@ARTICLE{Renner2023b,
       author = {{Renner}, Martin J. and {Quintino}, Marco T{\'u}lio},
        title = "{The minimal communication cost for simulating entangled qubits}",
      journal = {Quantum},
         year = 2023,
        month = oct,
       volume = {7},
        pages = {1149},
          doi = {10.22331/q-2023-10-24-1149},
archivePrefix = {arXiv},
       eprint = {2207.12457},
 primaryClass = {quant-ph},
       adsurl = {https://ui.adsabs.harvard.edu/abs/2023Quant...7.1149R}
}

@ARTICLE{Musin2010,
       author = {{Musin}, Oleg and {Tarasov}, Alexey},
        title = "{The strong thirteen spheres problem}",
      journal = {Discrete \& Computational Geometry },
         year = 2012,
        month = feb,
       volume = {48},
       number = {1},
        pages = {128–141},
          doi = {10.1007/s00454-011-9392-2},
archivePrefix = {arXiv},
       eprint = {1002.1439},
 primaryClass = {math.MG},
       adsurl = {https://ui.adsabs.harvard.edu/abs/2010arXiv1002.1439M}
}

@ARTICLE{Caves2002,
       author = {{Caves}, Carlton M. and {Fuchs}, Christopher A. and {Schack}, R{\"u}diger},
        title = "{Conditions for compatibility of quantum-state assignments}",
      journal = {Physical Review A},
         year = 2002,
        month = dec,
       volume = {66},
       number = {6},
          eid = {062111},
        pages = {062111},
          doi = {10.1103/PhysRevA.66.062111},
archivePrefix = {arXiv},
       eprint = {quant-ph/0206110},
}

@ARTICLE{Johnston2025,
       author = {{Johnston}, Nathaniel and {Russo}, Vincent and {Sikora}, Jamie},
        title = "{Tight bounds for antidistinguishability and circulant sets of pure quantum states}",
      journal = {Quantum},
         year = 2025,
        month = feb,
       volume = {9},
        pages = {1622},
          doi = {10.22331/q-2025-02-04-1622},
archivePrefix = {arXiv},
       eprint = {2311.17047},
 primaryClass = {quant-ph},
       adsurl = {https://ui.adsabs.harvard.edu/abs/2025Quant...9.1622J}
}

@ARTICLE{Russo2023,
       author = {{Russo}, Vincent and {Sikora}, Jamie},
        title = "{Inner products of pure states and their antidistinguishability}",
      journal = {Physical Review A},
         year = 2023,
        month = mar,
       volume = {107},
       number = {3},
          eid = {L030202},
        pages = {L030202},
          doi = {10.1103/PhysRevA.107.L030202},
archivePrefix = {arXiv},
       eprint = {2206.08313},
 primaryClass = {quant-ph},
       adsurl = {https://ui.adsabs.harvard.edu/abs/2023PhRvA.107c0202R}
}

@ARTICLE{Perry2015,
       author = {{Perry}, Christopher and {Jain}, Rahul and {Oppenheim}, Jonathan},
        title = "{Communication Tasks with Infinite Quantum-Classical Separation}",
      journal = {Physical Review Letters},
         year = 2015,
        month = jul,
       volume = {115},
       number = {3},
          eid = {030504},
        pages = {030504},
          doi = {10.1103/PhysRevLett.115.030504},
archivePrefix = {arXiv},
       eprint = {1407.8217},
 primaryClass = {quant-ph},
       adsurl = {https://ui.adsabs.harvard.edu/abs/2015PhRvL.115c0504P}
}

@ARTICLE{Jain2014,
       author = {{Bandyopadhyay}, Somshubhro and {Jain}, Rahul and {Oppenheim}, Jonathan and {Perry}, Christopher},
        title = "{Conclusive exclusion of quantum states}",
      journal = {Physical Review A},
         year = 2014,
        month = feb,
       volume = {89},
       number = {2},
          eid = {022336},
        pages = {022336},
          doi = {10.1103/PhysRevA.89.022336},
archivePrefix = {arXiv},
       eprint = {1306.4683},
 primaryClass = {quant-ph},
       adsurl = {https://ui.adsabs.harvard.edu/abs/2014PhRvA..89b2336B}
}

@ARTICLE{Heinosaari2019,
       author = {{Heinosaari}, Teiko and {Kerppo}, Oskari},
        title = "{Communication of partial ignorance with qubits}",
      journal = {Journal of Physics A Mathematical General},
         year = 2019,
        month = sep,
       volume = {52},
       number = {39},
          eid = {395301},
        pages = {395301},
          doi = {10.1088/1751-8121/ab3ae4},
archivePrefix = {arXiv},
       eprint = {1903.04899},
 primaryClass = {quant-ph},
       adsurl = {https://ui.adsabs.harvard.edu/abs/2019JPhA...52M5301H}
}

@ARTICLE{Heinosaari2018,
       author = {{Heinosaari}, Teiko and {Kerppo}, Oskari},
        title = "{Antidistinguishability of pure quantum states}",
      journal = {Journal of Physics A Mathematical General},
         year = 2018,
        month = sep,
       volume = {51},
       number = {36},
          eid = {365303},
        pages = {365303},
          doi = {10.1088/1751-8121/aad1fc},
archivePrefix = {arXiv},
       eprint = {1804.10457},
 primaryClass = {quant-ph},
       adsurl = {https://ui.adsabs.harvard.edu/abs/2018JPhA...51J5303H}
}

@article{KochenSpecker1967,
 ISSN = {00959057, 19435274},
 URL = {http://www.jstor.org/stable/24902153},
 author = {Simon Kochen and Ernst~P. Specker},
 journal = {Journal of Mathematics and Mechanics},
 number = {1},
 pages = {59--87},
 publisher = {Indiana University Mathematics Department},
 title = {The Problem of Hidden Variables in Quantum Mechanics},
 urldate = {2023-03-03},
 volume = {17},
 year = {1967},
 month = {July}
}

@ARTICLE{Montina2011KSmodel,
       author = {{Montina}, A.},
        title = "{Approximate simulation of entanglement with a linear cost of communication}",
      journal = {Physical Review A},
         year = 2011,
        month = oct,
       volume = {84},
       number = {4},
          eid = {042307},
        pages = {042307},
          doi = {10.1103/PhysRevA.84.042307},
archivePrefix = {arXiv},
       eprint = {1107.4647},
 primaryClass = {quant-ph},
       adsurl = {https://ui.adsabs.harvard.edu/abs/2011PhRvA..84d2307M}
}

@ARTICLE{Bae2025,
       author = {{Bae}, Joonwoo and {Flatt}, Kieran and {Heinosaari}, Teiko and {Kerppo}, Oskari and {Mohan}, Karthik and {Mu{\~n}oz-Moller}, Andr{\'e}s and {Rai}, Ashutosh},
        title = "{Random exclusion codes: Quantum advantages of single-shot communication}",
      journal = {Physical Review Research},
         year = 2026,
        month = feb,
       volume = {8},
       number = {1},
          eid = {013171},
        pages = {013171},
          doi = {10.1103/196j-lmzl},
archivePrefix = {arXiv},
       eprint = {2506.07701},
 primaryClass = {quant-ph},
       adsurl = {https://ui.adsabs.harvard.edu/abs/2026PhRvR...8a3171B}
}

@ARTICLE{Havlicek2020,
       author = {{Havl{\'\i}{\v{c}}ek}, Vojt{\v{e}}ch and {Barrett}, Jonathan},
        title = "{Simple communication complexity separation from quantum state antidistinguishability}",
      journal = {Physical Review Research},
         year = 2020,
        month = mar,
       volume = {2},
       number = {1},
          eid = {013326},
        pages = {013326},
          doi = {10.1103/PhysRevResearch.2.013326},
archivePrefix = {arXiv},
       eprint = {1911.01927},
 primaryClass = {quant-ph},
       adsurl = {https://ui.adsabs.harvard.edu/abs/2020PhRvR...2a3326H}
}

@article{Muller1959,
author = {Muller, Mervin E.},
title = {A note on a method for generating points uniformly on n-dimensional spheres},
year = {1959},
issue_date = {April 1959},
publisher = {Association for Computing Machinery},
address = {New York, NY, USA},
volume = {2},
number = {4},
issn = {0001-0782},
url = {https://doi.org/10.1145/377939.377946},
doi = {10.1145/377939.377946},
journal = {Commun. ACM},
month = apr,
pages = {19–20}
}

@ARTICLE{Montina2011,
       author = {{Montina}, Alberto},
        title = "{Communication cost of classically simulating a quantum channel with subsequent rank-1 projective measurement}",
      journal = {Physical Review A},
         year = 2011,
        month = dec,
       volume = {84},
       number = {6},
          eid = {060303},
        pages = {060303},
          doi = {10.1103/PhysRevA.84.060303},
archivePrefix = {arXiv},
       eprint = {1110.5944},
 primaryClass = {quant-ph},
       adsurl = {https://ui.adsabs.harvard.edu/abs/2011PhRvA..84f0303M}
}

@ARTICLE{Chakraborty2023,
       author = {{Chakraborty}, Ananya and {Gopalkrishna Naik}, Sahil and {Lobo}, Edwin Peter and {Krishna Patra}, Ram and {Sen}, Samrat and {Alimuddin}, Mir and {Mukherjee}, Amit and {Banik}, Manik},
        title = "{Overcoming Traditional No-Go Theorems: Quantum Advantage in Multiple Access Channels}",
      journal = { },
         year = 2023,
        month = sep,
archivePrefix = {arXiv},
       eprint = {2309.17263},
 primaryClass = {quant-ph},
}

@ARTICLE{FrenkelWeiner2015,
       author = {{Frenkel}, P{\'e}ter E. and {Weiner}, Mih{\'a}ly},
        title = "{Classical Information Storage in an n-Level Quantum System}",
      journal = {Communications in Mathematical Physics},
         year = 2015,
        month = {December},
       volume = {340},
       number = {2},
        pages = {563-574},
          doi = {10.1007/s00220-015-2463-0},
archivePrefix = {arXiv},
       eprint = {1304.5723},
 primaryClass = {cs.IT},
       adsurl = {https://ui.adsabs.harvard.edu/abs/2015CMaPh.340..563F}
}

@ARTICLE{Listspherecode,
       author = {Cohn, Henry},
        journal = "{Table of spherical codes}",
 NOTE =     "\url{https://hdl.handle.net/1721.1/153543} and 
             \url{https://spherical-codes.org}",
}

@ARTICLE{Robinson1961,
       author = {Raphael M. Robinson},
        title = "{Arrangement of 24 points on a sphere}",
      journal = {Mathematische Annalen},
         year = 1961,
        month = {February},
       volume = {144},
        pages = {17–48},
          doi = {10.1007/BF01396539}
}

@ARTICLE{Doolittle2026,
       author = {{Doolittle}, Brian and {Leditzky}, Felix and {Chitambar}, Eric},
        title = "{An Operational Framework for Nonclassicality in Quantum Communication Networks}",
      journal = {Quantum},
         year = 2026,
        month = apr,
       volume = {10},
        pages = {2052},
          doi = {10.22331/q-2026-04-08-2052},
archivePrefix = {arXiv},
       eprint = {2403.02988},
 primaryClass = {quant-ph},
       adsurl = {https://ui.adsabs.harvard.edu/abs/2026Quant..10.2052D}
}

@ARTICLE{Bowles2015b,
       author = {{Bowles}, Joseph and {Brunner}, Nicolas and {Paw{\l}owski}, Marcin},
        title = "{Testing dimension and nonclassicality in communication networks}",
      journal = {Physical Review A},
         year = 2015,
        month = aug,
       volume = {92},
       number = {2},
          eid = {022351},
        pages = {022351},
          doi = {10.1103/PhysRevA.92.022351},
archivePrefix = {arXiv},
       eprint = {1505.01736},
 primaryClass = {quant-ph},
       adsurl = {https://ui.adsabs.harvard.edu/abs/2015PhRvA..92b2351B}
}

@ARTICLE{Schlosser2026,
       author = {{Schl{\"o}sser}, Sebastian and {Kleinmann}, Matthias},
        title = "{Bounding the classical cost of simulating quantum behaviors in the prepare-and-measure scenario}",
    journal = { },
         year = 2026,
        month = mar,
archivePrefix = {arXiv},
       eprint = {2603.01255},
 primaryClass = {quant-ph},
}

@ARTICLE{Pandit2026,
       author = {{Pandit}, Ankush and {Hazra}, Soumyabrata and {Manna}, Satyaki and {Chaturvedi}, Anubhav and {Saha}, Debashis},
        title = "{Limits of classical correlations and quantum advantages under (anti-)distinguishability constraints in multipartite communication}",
      journal = {Physical Review A},
         year = 2026,
        month = mar,
       volume = {113},
       number = {3},
          eid = {032433},
        pages = {032433},
          doi = {10.1103/j54v-ng1h},
archivePrefix = {arXiv},
       eprint = {2506.07699},
 primaryClass = {quant-ph},
       adsurl = {https://ui.adsabs.harvard.edu/abs/2026PhRvA.113c2433P}
}

@ARTICLE{Buhrman2001PRL,
       author = {{Buhrman}, Harry and {Cleve}, Richard and {Watrous}, John and {de Wolf}, Ronald},
        title = "{Quantum Fingerprinting}",
      journal = {Physical Review Letters},
         year = 2001,
        month = oct,
       volume = {87},
       number = {16},
          eid = {167902},
        pages = {167902},
          doi = {10.1103/PhysRevLett.87.167902},
archivePrefix = {arXiv},
       eprint = {quant-ph/0102001},
 primaryClass = {quant-ph},
       adsurl = {https://ui.adsabs.harvard.edu/abs/2001PhRvL..87p7902B}
}

@inproceedings{Gavinsky2007,
  author    = {Gavinsky, Dmitry and Kempe, Julia and Kerenidis, Iordanis and Raz, Ran and de Wolf, Ronald},
  title     = {Exponential Separations for One-Way Quantum Communication Complexity, with Applications to Cryptography},
  booktitle = {Proceedings of the Thirty-Ninth Annual ACM Symposium on Theory of Computing},
  series    = {STOC '07},
  pages     = {516--525},
  year      = {2007},
  publisher = {Association for Computing Machinery},
  address   = {New York, NY, USA},
  doi       = {10.1145/1250790.1250866},
  eprint    = {quant-ph/0611209},
  archivePrefix = {arXiv},
  url       = {https://arxiv.org/abs/quant-ph/0611209}
}

@ARTICLE{sidajaya23,
       author = {{Sidajaya}, Peter and {Lim}, Aloysius Dewen and {Yu}, Baichu and {Scarani}, Valerio},
        title = "{Neural Network Approach to the Simulation of Entangled States with One Bit of Communication}",
      journal = {Quantum},
         year = 2023,
        month = {October},
       volume = {7},
        pages = {1150},
          doi = {10.22331/q-2023-10-24-1150},
archivePrefix = {arXiv},
       eprint = {2305.19935},
 primaryClass = {quant-ph},
       adsurl = {https://ui.adsabs.harvard.edu/abs/2023Quant...7.1150S}
}

@ARTICLE{Naik2025,
    author = {Gopalkrishna Naik, Sahil and Zartab, Mani and Gisin, Nicolas and Banik, Manik},
    title = {No-go theorem for generic simulation of qubit channels with finite classical resources},
    journal = {Proceedings of the Royal Society A: Mathematical, Physical and Engineering Sciences},
    volume = {482},
    number = {2333},
    pages = {20250831},
    year = {2026},
    month = {03},
    issn = {1364-5021},
    doi = {10.1098/rspa.2025.0831},
    url = {https://doi.org/10.1098/rspa.2025.0831},
archivePrefix = {arXiv},
       eprint = {2501.15807},
 primaryClass = {quant-ph},
}

@book{ConwaySloane1999,
  author    = {Conway, John H. and Sloane, Neil J. A.},
  title     = {Sphere Packings, Lattices and Groups},
  series    = {Grundlehren der mathematischen Wissenschaften},
  volume    = {290},
  edition   = {3rd},
  publisher = {Springer-Verlag},
  address   = {New York},
  year      = {1999}
}

@article{Roy2014,
author = {Roy, Aidan and Suda, Sho},
title = {Complex Spherical Designs and Codes},
journal = {Journal of Combinatorial Designs},
volume = {22},
number = {3},
pages = {105-148},
doi = {https://doi.org/10.1002/jcd.21379},
url = {https://onlinelibrary.wiley.com/doi/abs/10.1002/jcd.21379},
archivePrefix = {arXiv},
       eprint = {1104.4692},
 primaryClass = {math.CO},
year = {2014}
}

@article{DelsarteGoethalsSeidel1977,
  author  = {Delsarte, P. and Goethals, J. M. and Seidel, J. J.},
  title   = {Spherical codes and designs},
  journal = {Geometriae Dedicata},
  volume  = {6},
  number  = {3},
  pages   = {363--388},
  year    = {1977}
}

@article{Sloanlist,
  author  = "{N. J. A. Sloane, with the collaboration of R. H. Hardin, W. D. Smith and others}",
  title   = {Tables of Spherical Codes},
  journal = {published electronically at NeilSloane.com/packings/},
  url = {http://NeilSloane.com/packings}
}

@article{Naik2022,
  title = {Composition of Multipartite Quantum Systems: Perspective from Timelike Paradigm},
  author = {Naik, Sahil Gopalkrishna and Lobo, Edwin Peter and Sen, Samrat and Patra, Ram Krishna and Alimuddin, Mir and Guha, Tamal and Bhattacharya, Some Sankar and Banik, Manik},
  journal = {Physical Review Letters},
  volume = {128},
  issue = {14},
  pages = {140401},
  numpages = {7},
  year = {2022},
  month = {Apr},
  publisher = {American Physical Society},
  doi = {10.1103/PhysRevLett.128.140401},
  url = {https://link.aps.org/doi/10.1103/PhysRevLett.128.140401},
archivePrefix = {arXiv},
       eprint = {2107.08675},
 primaryClass = {quant-ph},
}

@ARTICLE{Chakraborty2026b,
       author = {{Chakraborty}, Ananya and {Banik}, Manik and {de Wolf}, Ronald},
        title = "{Exponential Advantage of Multipartite Entanglement over Quantum Communication with Applications to Bounded-Storage Cryptography}",
      journal = { },
         year = 2026,
        month = jul,
archivePrefix = {arXiv},
       eprint = {2607.27957},
 primaryClass = {quant-ph},
}

%%%%%%%%%%%%%%%%%%%%%%%%%%%%%%%%%%%%%%%%%%%%%%%%%%%%%%%%%%%%%%%%%%%%%%%%%%%%%%%%%%%%%
%%%%%%%%%%%%%%%%%%%%%%%%%%%%%%%% START OF APPENDIX %%%%%%%%%%%%%%%%%%%%%%%%%%%%%%%%%%
%%%%%%%%%%%%%%%%%%%%%%%%%%%%%%%%%%%%%%%%%%%%%%%%%%%%%%%%%%%%%%%%%%%%%%%%%%%%%%%%%%%%%

\onecolumngrid
\appendix

\section{Detailed quantum strategies}\label{quantumstrategies}

It is possible to generalise the task to an arbitrary number of parties. Suppose there are $n$ Alices and each receives an input $x_i\in \{1,2,\ldots,m\}$. Given that input, Alice-$i$ prepares a $d$-dimensional state $\ket{\Psi_{x_i}}$ and sends it to Bob. These states satisfy:
\begin{align}
    \forall\, j\neq k\in \{1,2,\ldots,m\}:\quad |\braket{\Psi_{j}}{\Psi_{k}}|\leq f(n):=\frac{1-(\sqrt[n]{2}-1)^2}{1+(\sqrt[n]{2}-1)^2}, \label{sphercode}
\end{align}
where $f(n)$ is defined as in~\eqref{deffunction}. We now describe the complete quantum strategy.

\subsection{Step 1: Reducing the angle}

Bob first applies a unitary transformation (mapping the two states into a qubit subspace) on the received state from Alice-$i$, satisfying:
\begin{align}
    U_i\ket*{\Psi_{\tilde{x}^0_i}}&=\cos{(\alpha_i)}\ket{0}+\sin{(\alpha_i)}\ket{1}\\
    U_i\ket*{\Psi_{\tilde{x}^1_i}}&=\cos{(\alpha_i)}\ket{0}-\sin{(\alpha_i)}\ket{1} \, ,
\end{align}
for some angle $\alpha_i$. Since unitary transformations preserve the inner product and $|\braket*{\Psi_{\tilde{x}^0_i}}{\Psi_{\tilde{x}^1_i}}|\leq f(n)$ by~\eqref{sphericalcode}, we know that $|\braket*{U_i\Psi_{\tilde{x}^0_i}}{U_i\Psi_{\tilde{x}^1_i}}|=\cos{(2\alpha_i)}\leq f(n)$. Since $\cos(2\alpha_i)\leq f(n)=\cos(2\alpha^c_n)$, it follows that $\alpha_i\geq \alpha^c_n$. For simplicity, we drop the index $i$ from now on.

In a next step, Bob prepares an ancilla qubit in $\ket{0}$ and applies the following unitary on the joint system:
\begin{align}
    \tilde{U}=\begin{pmatrix}
        \frac{\cos{(\alpha^c_n)}\cdot \cos{(\beta)}}{\cos{(\alpha)}} & 0 & 0 & -\frac{\sin{(\alpha^c_n)}\cdot \sin{(\beta)}}{\cos{(\alpha)}}\\
        0 & \frac{\sin{(\alpha^c_n)}\cdot \cos{(\beta)}}{\sin{(\alpha)}} & -\frac{\cos{(\alpha^c_n)}\cdot \sin{(\beta)}}{\sin{(\alpha)}} & 0 \\
        0 & \frac{\cos{(\alpha^c_n)}\cdot \sin{(\beta)}}{\sin{(\alpha)}} & \frac{\sin{(\alpha^c_n)}\cdot \cos{(\beta)}}{\sin{(\alpha)}} & 0  \\
        \frac{\sin{(\alpha^c_n)}\cdot \sin{(\beta)}}{\cos{(\alpha)}} & 0 & 0 & \frac{\cos{(\alpha^c_n)}\cdot \cos{(\beta)}}{\cos{(\alpha)}} \\
    \end{pmatrix},
\end{align}
where $\beta$ is chosen such that:
\begin{align}
    \cos(2\beta):=\frac{\cos{(2\alpha)}}{\cos{(2\alpha^c_n)}}.
\end{align}
Note that this equation has a solution since $\cos{(2\alpha)}\leq \cos{(2\alpha^c_n)}$. (Note also that the unitary is ill-defined if either $\cos{(\alpha)}=0$, $\sin{(\alpha)}=0$, or $\cos{(2\alpha^c_n)}=0$. None of these cases appear since $|\braket{\Psi_j}{\Psi_k}|\leq f(n)=\cos{(2\alpha^c_n)}<1$.) One can verify that $\tilde{U}$ is unitary and maps:
\begin{align}
    \tilde{U}\ket{0}\otimes (\cos{(\alpha)}\ket{0}+\sin{(\alpha)}\ket{1})=&(\cos{(\beta)}\ket{0}+\sin{(\beta)}\ket{1})\otimes(\cos{(\alpha^c_n)}\ket{0}+\sin{(\alpha^c_n)}\ket{1}),\\
    \tilde{U}\ket{0}\otimes (\cos{(\alpha)}\ket{0}-\sin{(\alpha)}\ket{1})=&(\cos{(\beta)}\ket{0}-\sin{(\beta)}\ket{1})\otimes(\cos{(\alpha^c_n)}\ket{0}-\sin{(\alpha^c_n)}\ket{1}).
\end{align}
After tracing out the ancilla qubit, Bob holds either $\ket{\psi_0}=\cos{(\alpha^c_n)}\ket{0}+\sin{(\alpha^c_n)}\ket{1}$ or $\ket{\psi_1}=\cos{(\alpha^c_n)}\ket{0}-\sin{(\alpha^c_n)}\ket{1}$.

We verify that $\tilde{U}$ is unitary, i.e.\ $\tilde{U}^\dagger \tilde{U}=\id$. All off-diagonal terms of $\tilde{U}^\dagger \tilde{U}$ vanish. For the diagonal terms, the first and fourth entries give:
\begin{align}
    \frac{\cos^2{(\alpha^c_n)}\cos^2{(\beta)}+\sin^2{(\alpha^c_n)}\sin^2{(\beta)}}{\cos^2{(\alpha)}}=1 \, .
\end{align}
To see this, we use that (note $\cos^2(x)=(1+\cos(2x))/2$ and $\sin^2(x)=(1-\cos(2x))/2$):
\begin{align}
\cos^2(\alpha^c_n)\cos^2(\beta)+\sin^2(\alpha^c_n)\sin^2(\beta)&=\frac{(1+\cos(2\alpha^c_n))(1+\cos(2\beta))}{4}+\frac{(1-\cos(2\alpha^c_n))(1-\cos(2\beta))}{4}\\
&=\frac{1+\cos(2\alpha^c_n)\cos(2\beta)}{2}=\frac{1+\cos(2\alpha)}{2}=\cos^2(\alpha)
\end{align}
An analogous calculation handles the second and third diagonal entries.

\subsection{Step 2: Applying the joint measurement}

After Step~1, Bob holds a qubit from each Alice, which is either $\ket{\psi_0}$ or $\ket{\psi_1}$. %Writing $c_n=\cos(\alpha^c_n)$ and $s_n=\sin(\alpha^c_n)$, 
The total state is one of the $2^n$ possibilities labelled by $\vec{r}\in\{0,1\}^n$:
\begin{align}
    \ket{\Omega_{\vec{r}}}:=&\ket{\psi_{r_1}}\otimes \ket{\psi_{r_2}} \otimes ... \otimes \ket{\psi_{r_n}}=\sum_{\vec{z}} (-1)^{\vec{r}\cdot \vec{z}}\cos(\alpha^c_n)^{(n-\vec{z}\cdot \vec{1})}\sin(\alpha^c_n)^{\vec{z}\cdot \vec{1}}\ket{\vec{z}}
\end{align}
where $\vec{1}=(1,1,...,1)$ is a vector where each of the $n$ entries equals one. Bob measures in the basis:
\begin{align}
    \ket{\Phi_{\vec{r}}}=\frac{1}{\sqrt{2^n}}\left( \ket*{\vec{0}} -\sum_{\vec{z}\neq \vec{0}} (-1)^{\vec{r}\cdot \vec{z}} \ket{\vec{z}} \right).
\end{align}
This forms an orthonormal basis since:
\begin{align}
    \braket{\Phi_{\vec{r}}}{\Phi_{\vec{s}}}&=\frac{1}{2^n}\sum_{\vec{z}}(-1)^{(\vec{r}+\vec{s})\cdot \vec{z}}=\delta_{\vec{r},\vec{s}},
\end{align}
where the last equality follows because the sum equals $2^n$ when $\vec{r}=\vec{s}$ and vanishes otherwise (terms pair up with opposite signs for any differing index). Now we verify $\braket{\Phi_{\vec{r}}}{\Omega_{\vec{r}}}=0$:
\begin{align}
    \sqrt{2^n}\braket{\Omega_{\vec{r}}}{\Phi_{\vec{r}}}=&\cos{(\alpha_n^c)}^n-\sum_{\vec{z}\neq \vec{0}}(-1)^{2\ \vec{r}\cdot \vec{z}}\cos(\alpha_n^c)^{(n-\vec{z}\cdot \vec{1})}\sin(\alpha^c_n)^{\vec{z}\cdot \vec{1}}\\
    =&\cos{(\alpha_n^c)}^n-\sum_{k=1}^{n}{n\choose k}\cos(\alpha_n^c)^{(n-k)}\sin(\alpha_n^c)^{k}\\
    =&2\cos{(\alpha_n^c)}^n-\sum_{k=0}^{n}{n\choose k}\cos(\alpha_n^c)^{(n-k)}\sin(\alpha_n^c)^{k}\\
    =&2\cos{(\alpha_n^c)}^n-(\cos(\alpha_n^c)+\sin(\alpha_n^c))^{n}\\
    =&\cos{(\alpha_n^c)}^n(2-(1+\tan{(\alpha_n^c)})^{n})\\
    =&0 \, ,
\end{align}
since $\tan{(\alpha_n^c)}=\sqrt[n]{2}-1$ (note that $\cos(2x)=(1-\tan^2(x))/(1+\tan^2(x))$ and compare with Eq.~\eqref{deffunction}) was chosen precisely to make this expression vanish. Hence, outcome $\ket{\Phi_{\vec{r}}}$ certifies that the preparation was not $\ket{\Omega_{\vec{r}}}$, and Bob can solve the task perfectly.

\section{An alternative approach}\label{app:johnston}
It might be worth mentioning that a weaker but simpler to state bound for $f(n)$ can be established via the results of Ref.~\cite{Johnston2025}. In fact, they show that: 
\begin{theorem}[{\cite[Cor.~5.3]{Johnston2025}}]\label{thmJohnston}
    A set of $k$ pure quantum states $\{\ket{ \psi_1}, \ket{ \psi_2}, \dots, \ket{ \psi_k}\}$ is antidistinguishable if
    \begin{align}\label{eq:johnston_condition}
        \norm{G}_F:=\sqrt{\sum_{i=1}^k \sum_{j=1}^k \lvert\braket{ \psi_i}{ \psi_j}\rvert^2}\leq \frac{k}{\sqrt{2}} \, .
    \end{align}
    Here, $G_{ij}=\braket*{ \psi_i}{ \psi_j}$ is the Gram matrix associated with the set of states $\{\ket{ \psi_i}\}_{i\in[k]}$, and $\norm{G}_F$ is its Frobenius norm.
\end{theorem}
Due to the tensor product structure this is easy to evaluate for the $k=2^n$ states in Eq.~\eqref{eq:prodstates} (namely, $\ket*{\Psi_{\vec{r}}} := \ket*{\Psi_{\tilde{x}^{r_1}_1}} \otimes \ket*{\Psi_{\tilde{x}^{r_2}_2}} \otimes \cdots \otimes \ket*{\Psi_{\tilde{x}^{r_n}_n}}\, , \quad \vec{r} \in \{0,1\}^n$):
\begin{align}
\begin{split}
    (\norm{G}_F)^2:=\sum_{i=1}^k \sum_{j=1}^k \lvert\braket{ \Psi_i}{ \Psi_j}\rvert^2
    =\prod_{i=1}^n (2+2\lvert\braket*{ \Psi_{\tilde{x}^0_i}}{ \Psi_{\tilde{x}^1_i}}\rvert^2)
    \leq 2^n (1+f(n)^2)^n\, .
\end{split}
\end{align}
we conclude that if $2^n (1+f(n)^2)^n\leq k^2/2=2^{2n}/2$, which translates to $f(n)^2 \leq 2/\sqrt[n]{2}-1$, all sets of states are antidistinguishable. Note further that $f(n)$ can approach 1, in the limit where $n$ is large. This is a simple and easy to establish bound but it turns out that we can derive a better bound by using the techniques introduced in the main text. Note also, that in this approach, Bob's measurements are not explicit but can be decided via semidefinite programming (SDP)~\cite{Jain2014, Johnston2025}:
\begin{align}\label{eq:SDP}
    \min_{\{M_{i}\}_i}&\ \sum_i \mathrm{Tr}[M_i  \rho_{i}] \\ \text{subject to}&:\ \sum_i M_i=\id\,,\, M_i \geq 0  \, .
\end{align}

\section{Detailed classical strategies}\label{classicalstrategies}

\subsection{Proof of Theorem~\ref{thm:classical}}
\begin{theorem}
    If each party sends a classical message of fewer than $m$ symbols, then Task~\ref{task} cannot be won with certainty, even in the presence of shared randomness.
\end{theorem}
\begin{proof}

%We first give the detailed argument that a classical message of fewer than $m$ symbols cannot win Task~\ref{task} with certainty, even when the parties share unlimited randomness $\lambda$.
We want to show that each classical strategy of fewer than $m$ symbols necessarily leads to one input combination, in which the parties cannot win with certainty.

%Since the winning probability is linear in the strategy, 
We may absorb, without loss of generality, any local randomness of the parties into $\lambda$ and assume that, for each fixed value of $\lambda$, all encodings $c_i=f^\lambda_i(x_i)$ and Bob's decoding are deterministic. Fix such a $\lambda$. Alice-$i$'s encoding then maps her $m$ possible inputs to fewer than $m$ messages, so it cannot be injective: by the pigeonhole principle, there exist two inputs $x_i\neq \tilde{x}_i$ with $f^\lambda_i(x_i)=f^\lambda_i(\tilde{x}_i)$. Suppose Bob's promise is exactly $(\tilde{x}^0_i,\tilde{x}^1_i)=(x_i,\tilde{x}_i)$ for every Alice. Then Bob receives the same messages $c_1,\ldots,c_n$ for all $2^n$ combinations consistent with this promise, so his output $\vec{b}$ is one fixed combination among them, independent of which combination was actually prepared. Since $\vec{b}$ is itself consistent with the promise, the input $\vec{x}=\vec{b}$ is possible, and for that particular input they cannot win with certainty. Hence, for every fixed $\lambda$, there is an input on which the corresponding deterministic strategy fails. Since there are finitely many input/promise pairs, a strategy winning with probability one on all of them would require a single $\lambda$ to win on all of them simultaneously, contradicting the above.
%This rules out a perfect strategy. Suppose the winning probability were exactly one for every input and every promise. There are only finitely many such pairs, and for each of them the win indicator is a $\{0,1\}$-valued random variable whose expectation over $\lambda$ equals one; hence it equals one for almost every $\lambda$. Intersecting these finitely many full-measure sets, there would exist a single value of $\lambda$ winning on all input/promise pairs simultaneously, contradicting the previous paragraph. Therefore some input has winning probability strictly less than one, which proves Theorem~\ref{thm:classical}. 
\end{proof}

\subsection{Winning probabilities for a given message length}

We now analyse how well classical strategies with a given communication budget can perform. Consider the task with two Alices, each with input in $\{1,\ldots,6\}$, and each allowed to send two bits. A natural strategy is: $c=00$ for $x=1$; $c=01$ for $x=2$; $c=10$ for $x\in\{3,4\}$; $c=11$ for $x\in\{5,6\}$.

If Bob is told that $x_1\in\{3,4\}$ and $x_2\in\{5,6\}$, he receives $c_1=10$ and $c_2=11$ for all four combinations and must guess. His success probability is $3/4$ for this instance. The same failure occurs for the four promises in which both Alices map two inputs to the same message. In all other cases Bob can exclude one preparation with certainty.

This approach can be generalized to larger values of $m$. By symmetry, we may assume Alice sends message $c=k$ for inputs $n_1+\cdots+n_{k-1}+1$ through $n_1+\cdots+n_k$. The number of pairs where Bob cannot discriminate Alice-$i$'s input from her message is $\sum_{c:\, n_c\geq 2}\binom{n_c}{2}$, where $\sum_c n_c=m$. The winning probability is:
\begin{align}
    p_{\mathrm{win}}=1-\frac{1}{2^n}\prod_{i=1}^{n} \frac{\sum_{c:\,n_c\geq 2} \binom{n_c}{2}}{\binom{m}{2}}.
\end{align}
For the two-Alice, two-bit example, there are $\binom{6}{2}=15$ possible promises per Alice. The strategy above has $\binom{2}{2}=1$ bad pair per colliding message, arising for four of the $15^2=225$ combinations. The winning probability is:
\begin{align}
    \frac{1}{225}\left(221\cdot 1 + 4 \cdot \frac{3}{4}\right)=\frac{224}{225}<1 \, .
\end{align}
%One can verify this is optimal among all two-bit strategies. (needs argument)

\section{Spherical Caps on the Complex Unit Sphere}\label{appspherecode}

Let $z,v \in \mathbb{C}^d$ be unit vectors, and define the spherical cap around $v$ with opening angle $\theta$ as:
\begin{align}
    C(\theta) = \bigl\{\, z \in \mathbb{C}^d : \|z\|=1,\; |\langle z,v\rangle| \ge \cos(\theta) \,\bigr\}.
\end{align}
Our goal is to compute the relative area of such a cap compared to the whole sphere.

\begin{proposition}
Let $z \in \mathbb{C}^d$ be Haar-uniform on the unit sphere, and let $v \in \mathbb{C}^d$ be a fixed unit vector. Then $X := |\langle z,v\rangle|^2 \sim \mathrm{Beta}(1,d-1)$, with density $g_X(x) = (d-1)(1-x)^{d-2}$ for $0\leq x\leq 1$.
\end{proposition}

\begin{proof}
By unitary invariance of the Haar measure, take $v = e_1$. Then $X = |z_1|^2$. A standard method to generate a Haar-uniform vector $z$ is to take $y_j \sim \mathcal{CN}(0,1)$ i.i.d.\ and normalise $z = y/\|y\|$~\cite{Muller1959}. Writing $A = |y_1|^2 \sim \chi^2_2$ and $B = \sum_{j=2}^d |y_j|^2 \sim \chi^2_{2(d-1)}$, we have $X = A/(A+B)$. By the classical Beta--chi-squared relation with $\alpha=1$, $\beta=d-1$, this gives $X \sim \mathrm{Beta}(1,d-1)$.
\end{proof}

\begin{corollary}
For $0 \le \theta \le \pi/2$, the fraction of the complex unit sphere contained in $C(\theta)$ is:
\begin{align}
\Pr(z \in C(\theta)) = [1 - \cos^2(\theta)]^{d-1} = [\sin(\theta)]^{2(d-1)}.
\end{align}
\end{corollary}

Using $\theta=2\alpha^c_n$, $\tan(\alpha^c_n)=\sqrt[n]{2}-1$, and $\sin(2\arctan(x))=2x/(1+x^2)$:
\begin{align}
    \frac{S_{2\alpha^c_n}}{S}=\left( \frac{2(\sqrt[n]{2}-1)}{1+(\sqrt[n]{2}-1)^2}\right)^{2(d-1)},
\end{align}
so that:
\begin{align}
    \bar{m}(n,d)\geq \left( \frac{1+(\sqrt[n]{2}-1)^2}{2(\sqrt[n]{2}-1)}\right)^{2(d-1)}.
\end{align}
To obtain the simpler bound, note that:
\begin{align}
    \frac{1+(\sqrt[n]{2}-1)^2}{2(\sqrt[n]{2}-1)}=\frac{1}{2(\sqrt[n]{2}-1)}+\frac{(\sqrt[n]{2}-1)}{2}\geq \frac{n}{2},
\end{align}
where we use $1/(\sqrt[n]{2}-1)\geq n$. This follows from $2^x\leq 1+x$ for $0\leq x\leq 1$ and setting $x=1/n$ gives $\sqrt[n]{2}-1\leq 1/n$. Hence:
\begin{align}
    \bar{m}(n,d)\geq \left(\frac{n^2}{4}\right)^{d-1}.
\end{align}
%Asymptotically, $\bar{m}(n,d)\sim (n/\ln 2)^{2(d-1)}$ since $\lim_{n\to\infty} 1/((\sqrt[n]{2}-1)\cdot n) = 1/\ln 2$.

\end{document}